\documentclass{amsart}
\usepackage{mathrsfs,amsmath,amssymb,amsfonts}

\usepackage{tikz}
\usetikzlibrary{automata,positioning,arrows.meta}

\usepackage{todonotes}

\newcommand\e{\mathrm{e}}

\newtheorem{theorem}{Theorem}
\theoremstyle{plain}

\newtheorem{corollary}{Corollary}

\newtheorem{proposition}{Proposition}

\theoremstyle{definition}
\newtheorem{definition}{Definition}
\newtheorem{example}{Example}

\begin{document}
\title[]{Opacity complexity of automatic sequences.\\ The general case}
\author{Jean-Paul ALLOUCHE and Jia-Yan YAO}
\date{\today}
\keywords{finite automaton, automatic sequence, opacity, complexity}
    
\begin{abstract}
In this work we introduce a new notion called opacity complexity to measure the complexity of automatic 
sequences. We study basic properties of this notion, and exhibit an algorithm to compute it. As applications, 
we compute the opacity complexity of some well-known automatic sequences, including in particular constant 
sequences, purely periodic sequences, the Thue-Morse sequence, the period-doubling sequence,
the Golay-Shapiro(-Rudin) sequence, the paperfolding sequence, the Baum-Sweet sequence, the Tower 
of Hanoi sequence, and so on.

\medskip

\noindent
{\bf Keywords:} Opacity. Opacity complexity. Automatic sequences.
\end{abstract}

\maketitle

\begin{flushright}
    \thispagestyle{empty}
    \vspace{-5mm}
    {\sl Dedicated to the memory of Michel Mend\`es France}
    \vspace{2mm}
\end{flushright}

\markboth{Jean-Paul ALLOUCHE, and Jia-Yan YAO}{Opacity complexity of automatic sequences}%

\section{Introduction}

We begin with some definitions and notation (see for 
example \cite{eil} and \cite{AS}).

Let $A$ be a finite nonempty set. We call it an \emph{alphabet}. We let $\# A$ denote
the number of elements in $A$. Each element of $A$ is called a \emph{letter}. 
We fix $\varepsilon $ an element not in $A$, called the \emph{empty word} over $A$.

Let $\mathbb{N}=\{0,1,\ldots \}$ be the set of all natural numbers and let $%
n\in \mathbb{N}$. If $n = 0$, we define $A^{0}:=\{\varepsilon \}$, and if $n \neq 0$, we let 
$A^{n}$ denote the set of all finite sequences with terms in $A$ and of length $n$. Finally we set%
\begin{equation*}
A^{\ast }:=\bigcup\limits_{n=0}^{\infty }A^{n}\text{ and }\widehat{A}
:=A^{\ast }\cup A^{\mathbb{N}}.
\end{equation*}%
An element $w$ in $\widehat{A}$ is called a \emph{finite word} if $w\in A^{\ast }$ 
and an \emph{infinite word} if $w\in A^{\mathbb{N}}$. We let $|w|$ denote the
length of $w$. More precisely, we have $|w|=$ $n$ if $%
w\in A^{n}$, and $|w|=+\infty $ if $w\in A^{\mathbb{N}}$. In particular, we
have $|\varepsilon |=0$.

Let $w=(w(n))_{0\leqslant n<|w|}\in A^{\ast }$ and $v=(v(n))_{0\leqslant
n<|v|}\in \widehat{A}$ be two words over $A$. The \emph{concatenation} or
\emph{product} between $w$ and $v$, denoted by $w\ast v$ (or more simply by $wv$), is again a word of length $|w|+|v|$ over $A$, defined as follows:
\begin{equation*}
(w\ast v)(n)=\left\{
\begin{array}{cl}
w(n), & \mathrm{if}\text{ }0\leqslant n<|w|, \\
v(n-|w|), & \mathrm{if}\text{ }|w|\leqslant n<|w|+|v|.%
\end{array}%
\right.
\end{equation*}%
Therefore $w\varepsilon =\varepsilon w=w$ for all $w\in A^{\ast }$. Clearly $%
(A^{\ast },\ast )$ is a monoid with $\varepsilon $ as the identity element.
By induction, we can also define the concatenation of a finite or even an
infinite number of words over $A$. Thus every $w=(w(n))_{0\leqslant
n<|w|}\in \widehat{A}$ can be represented by a finite or an infinite product%
\begin{equation*}
w=\prod\limits_{n=0}^{|w|-1}w(n):=w(0)w(1)\cdots ,
\end{equation*}%
and every \emph{prefix} of $w$ can be written as $w[0,n] := w(0)\cdots w(n)$,
with $0\leqslant n<|w|$.

Let $w=(w(n))_{0\leqslant n<|w|}$ and $v=(v(n))_{0\leqslant n<|v|}$ be two
words over $A$. We define%
\begin{equation*}
\mathbf{d}_{A}(w,v)=2^{-\inf \{n\ :\ w(n) \neq v(n),\ 0 \leqslant n< \min (|w|,|v|)\}}
\end{equation*}%
if $w\neq v$, and $\mathbf{d}_{A}(w,v) = 0$ if $w = v$. Clearly $\mathbf{d}_{A}$
is a metric over $\widehat{A}$. Endowed with this metric, $\widehat{A}$ becomes a
compact metric space which contains $A^{\ast }$ as a dense subset. Finally we
remark that $A^{\mathbb{N}}$ is a compact subspace of $\widehat{A}$.
\medskip

From now on, we fix $k\geqslant 2$ an integer, and $\Sigma_k=\{0,1,\ldots, k-1\}$.
\medskip

A \emph{finite $k$-automaton} is a quadruple $\mathscr{A}=(S,i_0,\Sigma_k ,t)$\
which consists of

\begin{itemize}
\item an alphabet $S$ of states; one of the states, say $i_0$, is
distinguished and called the \emph{initial state}.

\item a map $t:S\times \Sigma_k \rightarrow S$, called the \emph{%
transition function}. \
\end{itemize}

\medskip

For all $s\in S$, put $t(s,\varepsilon )=s$. Then extend $t$ over $S\times\Sigma ^{\ast }_k$ 
(still denoted by $t$) such that $t(s,\sigma \eta):=t(t(s,\sigma ),\eta )$, for all $s\in S$ and all 
$\sigma ,\eta \in \Sigma^{\ast }_k$. The finite $k$-automaton $\mathscr{A}$ also induces a
map (also denoted by $\mathscr{A}$) from $\widehat{\Sigma }_k$ to $\widehat{S}$ defined by
\begin{equation*}
(\mathscr{A}\eta )(m):=t(i_0,\eta \lbrack 0,m])=t(i,\eta (0)\cdots \eta (m)),
\end{equation*}%
for all $\eta \in \widehat{\Sigma}_k$ and all $m\in \mathbb{N}$ ($0\leqslant
m<|\eta |$).

It is useful to give a pictorial representation of $\mathscr{A}=(S,i_0,\Sigma_k,t)$.
States will be represented by points or nodes or vertices. For all $s\in S$ and all $\sigma \in \Sigma_k$,
we link $s$ to $t(s,\sigma )$ by a (directed) arrow, labelled $\sigma $. This arrow (also called edge) is said
of \emph{type }$\sigma $ and denoted by $(s,\sigma ,t(s,\sigma ))$ (\emph{i.e.}, 
treated as an element in $S\times \Sigma_k \times S$) where $s$ is the
starting-point, $\sigma $ is the label or type of the arrow,
and $t(s,\sigma)$ is the endpoint. In this way, by linking sequentially all the edges together,
each $\eta \in \widehat{\Sigma}_k$ defines a path (noted $\mathfrak{p}_{\eta}$)
on the graph $\mathscr{A}$ (the path is infinite if $\eta \in \Sigma_k^{\mathbb{N}}$) as follows:
\begin{eqnarray}\label{eq0}
\mathfrak{p}_{\eta}=\big(i_0, \eta(0), t(i_0,\eta(0))\big)\big(t(i_0,\eta(0)),\eta(1), t(i_0,\eta(0)\eta(1))\big)\cdots.
\end{eqnarray}
In this work we only need consider paths of the above form, which begin from the initial state. Later in the next work \cite{allyao},
we shall be obliged to consider more general paths which can begin from any state.

Below we shall constantly identify $\mathscr{A}$ with its graph (and we use a horizontal incident arrow to mark
the initial state). Then $S$ becomes the set of vertices and $\Sigma_k $
becomes the set of labels or types of arrows. When we talk of a path,
we always mean that the path is finite unless otherwise indicated.

Let $r,s$ be two states of $\mathscr{A}=(S,i_0,\Sigma_k,t)$. We say that $s$ is
\emph{accessible from} $r$ if there exists $\sigma \in \Sigma_k ^{\ast }$
such that $s=t(r,\sigma )$. So $s$ is accessible from itself for 
$t(s,\varepsilon )=s$. A state of $\mathscr{A}$ is said \emph{accessible} if
it is accessible from the initial state $i_0$, and we
call $\mathscr{A}$ an \emph{accessible} (resp.~\emph{strictly accessible})
automaton if every state of $\mathscr{A}$ is accessible (resp.~for all
states $r$ and $s$, $r$ is accessible from $s$ and \emph{vice versa}). From
now on, all finite $k$-automata in discussion will be supposed (implicitly)
accessible, and we let $\mathrm{AUT}_k$ denote the set of all such finite $k$-automata.

Fix $Y$ a nonempty set. Let $o$ be a map defined on $S$ with values in $Y$. We shall
call the couple $(\mathscr{A},o)=(S,i_0,\Sigma_k,t,o)$ a\emph{\ finite $k$-automaton with output} and $o$ the \emph{output function} of $\mathscr{A}$. Just like the finite $k$-automaton $\mathscr{A}$, this
couple also induces a map (still denoted by $(\mathscr{A},o)$) from $\widehat{\Sigma }_k$ to $\widehat{o(S)}$ such that
\begin{equation*}
\forall\sigma \in \widehat{\Sigma}_k\ \text{and}\ \forall m\in\mathbb{N}\ (0\leqslant m<|\sigma |),\ 
\text{we have}\ (\mathscr{A},o)(\sigma)(m):=o((\mathscr{A}\sigma )(m)).
\end{equation*}%
Often, to simplify the notation, we also let $o(\mathscr{A}\sigma)$ denote $(\mathscr{A},o)(\sigma)$.

A sequence $u=(u(n))_{n\geqslant 0}$ with terms in $Y$ will be called a \emph{$k$-automatic
sequence }if there exists a finite $k$-automaton with output $(\mathscr{A},o)=(S,i_0,\Sigma _{k},t,o)$
such that $u(0)=o(i_0)$, and $u(n)=o(t(i_0,n_{m}\cdots n_{0}))$ for all integers $n\geqslant 1$ with standard
$k$-ary expansion $n=\sum\limits_{j=0}^{m}n_{j}k^{j}$. In this case, we also say
that $u$ is generated by $(\mathscr{A},o)$ (resp.~by $\mathscr{A}$ if $o$ is the identity map on $S$). Note that we can suppose in addition $t(i_0,0)=i_0$, and we call such an $\mathscr{A}$ an internal finite $k$-automaton of $u$. 

Indeed, if $u$ is generated by $(\mathscr{A},o)=(S,i_0,\Sigma _{k},t,o)$ 
with $t(i_0,0)\neq i_0$, then by adding a new state $i^{\prime }_0$ to 
$S$ and defining
\begin{eqnarray*}
&&S^{\prime}=S\cup \{i^{\prime }_0\},\\
&&t^{\prime }\big|_{S\times \Sigma _{k}}=t,\
t^{\prime }(i^{\prime}_0,0)=i^{\prime }_0,\ 
t^{\prime }(i^{\prime}_0,\sigma)=t(i_0,\sigma)\ (\forall \sigma\in \Sigma _{k}\setminus \{0\}),\\
&&o^{\prime }\big|_{S}=o,\ o^{\prime }(i^{\prime}_0)=o(i_0),
\end{eqnarray*}
we obtain a new finite 
$k$-automaton with output $(\mathscr{A}^{\prime },o^{\prime })=(S^{\prime},i^{\prime}_0,\Sigma _{k},t^{\prime },o^{\prime })$, which generates $u$ and satisfies $t'(i'_0,0)=i'_0$.

In what follows we shall let $\mathrm{AUT}_k(u)$ 
denote the set of all internal finite $k$-automata of $u$.
By the above discussion, we have $\mathrm{AUT}_k(u)\neq \emptyset$.  

Note here that in our definition of finite automata, we have not mentioned the notion of terminal state, that the preceding definition corresponds to the classical notion of complete
(deterministic) automaton where all states are final (cf. \cite{eil}), and that when we define the opacity of a finite automaton, we do not want to recognize a language, but only consider the finite automaton as a machine (transducer) which transforms a sequence into another one. For this point of view and related studies, see for example \cite{liardet,carton2,carton1}.

\medskip

Now we give two examples to illustrate the above 
definitions and notation.

\begin{example}
\label{ex1}{\em (One-state automaton)} Let $S=\{i_0\}.$ For all $\sigma \in \Sigma_k $,
put $t(i_0,\sigma )=i_0$. The finite $k$-automaton $\mathscr{I}_{k}=(S,i_0,\Sigma_k,t)$ is strictly accessible and generates the constant
sequence $i_0i_0i_0\cdots $.
\end{example}

\begin{example}
\label{ex2} {\em (Identity automaton)} Let $S=\{A,B\}$, $i_0=A$, $\Sigma_2=\{0,1\}$,
and define the transition function $t$ by $t(A,0)=A$, $t(B,0)=A$, $t(A,1)=B$, and $t(B,1)=B$. The finite $2$-automaton 
$\mathscr{A}_{id}=(S,i_0,\Sigma_2,t)$
is strictly accessible, and if we define $o(A)=0$ and $o(B)=1$, then $(\mathscr{A}_{id},o)(\eta )=\eta $, for all $\eta \in \widehat{\Sigma}_2$.
Moreover the $2$-automatic sequence generated by $\mathscr{A}_{id}$ is simply the
periodic sequence $ABAB\cdots .$%
\vskip -0.35cm

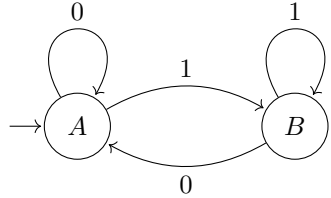
\begin{figure}[ht]
    
\begin{tikzpicture}[shorten >=1pt,node distance=2cm,auto, initial text=]
 
  \node[initial, state]  (A)              {$A$};
  \node[state]  (B) [right=of A] {$B$};

  
  \path[->] (A) edge [in=60,out=120,loop] node[above] {$0$} (A)
             (A) edge [bend left, above]   node {$1$} (B)
             (B) edge [bend left, below]   node {$0$} (A)
            (B) edge [in=60,out=120,loop] node[above] {$1$} (B) ; 
\end{tikzpicture}
\vskip -0.15cm

\centerline{Output function: $o(A) = 0$, $o(B) = 1$}
\vskip -0.4cm
\caption{Identity automaton $\mathscr{A}_{id}$} \label{iden}
\end{figure}

\end{example}

\section{Basic properties of opacity complexity}
In this section, we shall introduce the notion of 
opacity complexity of $k$-automatic sequences. We begin 
with the notion of opacity of finite automata. For this, 
we fix $Y$, a nonempty set.

\begin{definition}\label{def1} Let $\mathscr{A}=(S,i_0,\Sigma_k,t)$ be a finite $k$-automaton.  Define
\begin{eqnarray*}
\Omega_{\mathbf{d}} (\mathscr{A})=\sup_{\sigma \in \Sigma_k ^{\mathbb{N}}}\inf_{o\in Y ^{S}}\mathbf{d}(o(\mathscr{A}\sigma ),\sigma ),
\end{eqnarray*}
and call it the opacity of $\mathscr{A}$ attached to 
$\mathbf{d}$, where $\mathbf{d}$ is a ``prefixed'' comparison 
method (\emph{i.e.}, a comparison method where two elements having 
a long common prefix are ``close'') to measure the difference 
between the inputs and outputs.
\end{definition}

\noindent {\bf Remarks.} { 1. The comparison method $\mathbf{d}$ in the definition need be indicated in concrete 
problems. It is important to point out that $\mathbf{d}$ is crucial in the study of 
opacity theory: the theory may be quite different for different $\mathbf{d}$ (see \cite{yao,yao1}). Later in the present work 
we shall discuss in detail the case where $Y=\Sigma_k$ and  $\mathbf{d}=\mathbf{d}_Y$, while in the next work \cite{allyao},
we shall investigate the case where $Y=\mathbb{C}$ and $\mathbf{d}=\mathbf{d}_2$, defined for all $u\in \mathbb{C}^{\mathbb{N}}$
and all $\sigma \in \Sigma_k^{\mathbb{N}}$ by
\begin{eqnarray*}
\mathbf{d}_2(u,\sigma)=\limsup_{N\rightarrow\infty}\Big(\frac{1}{N}\sum_{m=0}^{N-1}\big|u(m)-\e(\sigma(m))\big|^2 \Big)^{1/2}.
\end{eqnarray*}
Here we use $\e(\sigma(m)):=\e^{\frac{2\sigma(m)\pi}{k}i}$ instead of $\sigma(m)$ to equilibrate the elements in $\Sigma_k$.
Example \ref{baum-sweet} shows that the two cases may be quite different for special situations. Indeed the opacity attached to $\mathbf{d}_{\Sigma_k}$ depends heavily on the starting homogeneous parts of paths issued from the initial state (see Theorem~\ref{thm2} below), while the opacity attached to $\mathbf{d}_2$ depends on simple circuits, {\em i.e.}, simple cyclic paths (see \cite{allyao} or \cite{yao2}).

2. The opacity of $\mathscr{A}$ measures in a certain sense the intrinsic noise produced by the default of $\mathscr{A}$.
Opacity theory of finite automata begun with M. Mend\`{e}s France in his
pioneer work \cite{men}, where he considered deterministic finite automata
with inputs $\pm $ and chose as comparison method the quadratic semi-norm.
It was then generalized and developed systematically by J.-Y. Yao to more
general deterministic finite automata in \cite{yao,yao1,yao2} (see also \cite{cheyao, huchen}) from the point of view of information transmission theory.

3. The term ``opacity'' has quite different meanings in information theory (mainly about security). See for example \cite{berard2015,laf2018,bar2021,hem2021,yang2021,win2022,Li2023,vil} and the references therein.}

\medskip 

Let $\mathscr{A}=(S,i_0,\Sigma_k ,t)$ and $\mathscr{A}^{\prime }=(S^{\prime},i^{\prime}_0,\Sigma_k ,t^{\prime })$ 
be two finite $k$-automata. We
call $\mathscr{A}^{\prime }$ a \emph{factor} of $\mathscr{A}$ (see, \emph{e.g.}, \cite{kam}) if there exists a {\em surjective} map $\lambda:S\rightarrow S^{\prime }$ such that 
$i^{\prime }_0=\lambda (i_0)$, and $t^{\prime }(\lambda (s),\sigma)=\lambda (t(s,\sigma ))$,
for all $s\in S$ and all $\sigma \in \Sigma_k$. In this case, we 
call $\lambda $ a \emph{$k$-automaton homomorphism of }$\mathscr{A}$, 
and write $\mathscr{A}^{\prime }=\lambda (\mathscr{A})$. From 
Definition~\ref{def1} we obtain at once
\begin{eqnarray}\label{eq1}
\Omega_{\mathbf{d}}(\mathscr{A})\leqslant\sup_{\sigma \in \Sigma_k ^{\mathbb{N}}}\inf_{o'\in Y ^{S'}}\mathbf{d}\big((o'\circ \lambda)(\mathscr{A}\sigma ),\sigma \big)
=\Omega_{\mathbf{d}} (\lambda (\mathscr{A})),
\end{eqnarray}
which means that {\em a finite $k$-automaton with simpler structure has a larger opacity.}
Note that the one-state $k$-automaton $\mathscr{I}_{k}$, which is the finite $k$-automaton with the simplest structure, is a factor of $\mathscr{A}$, hence $\Omega_{\mathbf{d}} (\mathscr{A})\leqslant \Omega_{\mathbf{d}} (\mathscr{I}_{k})$, from which we deduce
\begin{eqnarray*}
M_k:=\sup\limits_{\mathscr{A}\in \mathrm{AUT}_k}\Omega_{\mathbf{d}} (\mathscr{A})=\Omega_{\mathbf{d}} (\mathscr{I}_{k}).
\end{eqnarray*}

\begin{definition} Let $\mathscr{A}=(S,i_0,\Sigma_k, t)$ be a finite $k$-automaton.
\begin{itemize}
\item[{(1)}] We say that $\mathscr{A}$ is transparent if $\Omega_{\mathbf{d}} (\mathscr{A})=0$.

\item[{(2)}] We say that $\mathscr{A}$ is opaque if $\Omega_{\mathbf{d}} (\mathscr{A})=M_k$.
\end{itemize}
\end{definition}

Let $\lambda $ be a $k$-automaton homomorphism of $\mathscr{A}$. If $\lambda $ is also injective,
then its inverse map $\lambda ^{-1}$ is a $k$-automaton homomorphism of $\lambda (\mathscr{A})$, and we
call $\lambda $ a \emph{$k$-automaton isomorphism of} $\mathscr{A}$,
and say that $\mathscr{A}$ and $\mathscr{A}^{\prime }$ are \emph{isomorphic}, noted $\mathscr{A}\simeq \mathscr{A}^{\prime }$. In this case, we have $\Omega_{\mathbf{d}} (\mathscr{A})=\Omega_{\mathbf{d}} (\mathscr{A}')$.
Intuitively two finite $k$-automata are isomorphic if and only if, 
up to the names of states, they have the same graph. From now on, 
we shall always identify isomorphic finite $k$-automata and use, 
if no confusion is possible, the same symbols $\mathscr{A}, \mathscr{B}$, and so on, for finite $k$-automata and for classes 
of isomorphic finite $k$-automata. In particular, up to 
isomorphism, there exists only one one-state 
$k$-automaton $\mathscr{I}_k$.

\begin{definition} Let $u=(u(n))_{n\geqslant 0}$ be $k$-automatic with terms in $Y$.  We define
\begin{eqnarray*}
\Omega_{\mathbf{d}} (u)=\sup_{\mathscr{A}\in \mathrm{AUT}_k(u)}\Omega_{\mathbf{d}} (\mathscr{A}),\ \textrm{and}\ \varpi_{\mathbf{d}} (u)=\frac{\Omega_{\mathbf{d}} (u)}{M_k},
\end{eqnarray*}
and call them respectively the opacity of $u$ and the opacity complexity of $u$.
\begin{itemize}
\item[{(1)}] We say that $u$ is transparent if $\Omega_{\mathbf{d}} (u)=0$, \emph{i.e.}, 
$\varpi_{\mathbf{d}} (u)=0$;

\item[{(2)}] We say that $u$ is opaque if $\Omega_{\mathbf{d}} (u)=M_k$, \emph{i.e.}, 
$\varpi_{\mathbf{d}} (u)=1$.
\end{itemize}
\end{definition}

\noindent {\bf Remark.} For a given infinite word (sequence) $u$, there are two ways to measure its complexity.
One is based on the internal structure of $u$. A type example is the subword complexity function $p_u(n)$ which counts
the number of different subwords of length $n$ in $u$ (see for example \cite[p.\,298]{AS}). The other uses external tools to locate $u$ in the complexity hierarchy.
A typical example is the transducer degrees introduced in \cite[p.\,830]{en}, which compare two infinite words $w,u$, 
and write $w\geqslant u$, if $u$ is the image of $w$ under a sequential finite-state transducer. 
Our complexity belongs to the first type, for $\mathrm{AUT}_k(u)$ is uniquely determined by $u$ itself. Below we show in Theorem \ref{pr2} that the sup in the definition
can be achieved by a special finite $k$-automaton $\mathscr{A}_u$, treated as a sequential finite-state transducer. However there $\mathscr{A}_u$ is not an external tool to compare $u$ with the others, but a part of $u$ to define the complexity.

\medskip

Let $(\mathscr{A},o)=(S,i_0,\Sigma_k,t,o)$ be a finite $k$-automaton with
output. Two states $r,s$ of $\mathscr{A}$ are said \emph{indistinguishable}
if $o(t(r,\sigma ))=o(t(s,\sigma ))$ for all $\sigma \in \Sigma_k^*$.
Otherwise we call them \emph{distinguishable}. If all distinct states of $%
\mathscr{A}$ are distinguishable, then $(\mathscr{A},o)$ is called \emph{%
minimal}.

Let $(\mathscr{A},o)$ and $(\mathscr{A}^{\prime },o^{\prime })$ be two
finite $k$-automata with output. If for all $\eta \in \Sigma_k^*$, we have $(\mathscr{A},o)(\eta )=(\mathscr{A}^{\prime },o^{\prime })(\eta
)$, then we call them \emph{equivalent} and write $(\mathscr{A},o)\approx (\mathscr{A}^{\prime },o^{\prime })$. If in addition $\mathscr{A}\simeq
\mathscr{A}^{\prime }$, then we call them \emph{isomorphic} and write $(%
\mathscr{A},o)\cong (\mathscr{A}^{\prime },o^{\prime })$.
It is well known that two equivalent minimal $k$-automata are isomorphic, and every finite $k$-automaton with 
output is equivalent to some minimal $k$-automaton 
(see for example \cite{con}).
This result can be slightly specified by the following 
one, which shows in a certain sense that a minimal 
$k$-automaton is in fact the least common factor of all
finite $k$-automata with output, which are equivalent to 
it (see for example \cite[Proposition 9, p.\,262]{yao3}).

\begin{proposition}\label{pr1}
For each finite $k$-automaton with output $(\mathscr{A},o) $,
there exists a minimal $k$-automaton $(\mathscr{A}^{\prime },o^{\prime })$ (unique up to isomorphism)
such that $\mathscr{A}^{\prime }$ is a factor of $\mathscr{A}$, and $(\mathscr{A},o)\approx (\mathscr{A}^{\prime },o^{\prime })$.
\end{proposition}

Let $u=(u(n))_{n\geqslant 0}$ be $k$-automatic with terms in $Y$. Let $\mathscr{A}_u\in\mathrm{AUT}_k(u)$
denote the common factor of all  finite $k$-automata in $\mathrm{AUT}_k(u)$ (whose existence is guaranteed 
by Proposition \ref{pr1}), called the intrinsic finite $k$-automaton of $u$ (note here that it is also the finite $k$-automaton in $\mathrm{AUT}_k(u)$ with the smallest number of states). By Proposition~\ref{pr1} and  
Formula~(\ref{eq1}), we obtain immediately the following result which reduces the computation of the opacity 
complexity of $k$-automatic sequences to the computation of the opacity of finite $k$-automata.

\begin{theorem}\label{pr2}
Let $u=(u(n))_{n\geqslant 0}$ be $k$-automatic with terms in $Y$. Then
\begin{eqnarray*}
\Omega_{\mathbf{d}} (u)=\Omega_{\mathbf{d}} (\mathscr{A}_u),\ \textrm{and}\ \varpi_{\mathbf{d}} (u)=\frac{\Omega_{\mathbf{d}} (\mathscr{A}_u)}{M_k}.
\end{eqnarray*}
\end{theorem}

\noindent {\bf Remark.} {As it has been indicated above, the opacity of a finite $k$-automaton $\mathscr{A}$ measures in a certain sense the intrinsic noise produced by the default of $\mathscr{A}$. So the opacity complexity of a $k$-automatic sequence $u$
measures the intrinsic noise of its internal finite $k$-automaton with the simplest structure (\emph{i.e.}, its intrinsic finite $k$-automaton $\mathscr{A}_u$), and describes the distortion between $u$ and the words defined by the $k$-ary expansion of $n$'s. }

\section{Computation of opacity complexity}

In the section, we shall concentrate our attention on the case where $Y=\Sigma_k$, 
and $\mathbf{d}=\mathbf{d}_{\Sigma_k}$. In this case,
we already know how to compute the opacity of a finite $k$-automaton $\mathscr{A}$ from 
the structure of the graph of $\mathscr{A}$ (see \cite{yao1}, and also \cite{huchen}).
For the convenience of the potential readers, in the following we shall recall all the needed results.
 To simplify the notation, we shall also let $\Omega(u)$ and $\omega (u)$ 
 denote respectively $\Omega_{\mathbf{d}}(u)$ and $\omega_{\mathbf{d}}(u)$.
Then
\begin{eqnarray*}
M_k=\Omega(\mathscr{I}_{k})=\sup_{\sigma \in \Sigma_k ^{\mathbb{N}}}\inf_{o\in \Sigma_k^{S}}\mathbf{d}(o(\mathscr{I}_{k}\sigma ),\sigma )=\frac{1}{2},
\end{eqnarray*}
where the supremum is attained for all  $\sigma \in \Sigma_k ^{\mathbb{N}}$ with $\sigma(0)\neq \sigma(1)$.

Let $\mathscr{A}=(S,i_0,\Sigma_k, t)$ be a finite $k$-automaton, and $\mathfrak{p}$ a path on $\mathscr{A}$. 
We shall say that $\mathfrak{p}$
is $\emph{homogeneous}$ if for each vertex $s$ of $\mathfrak{p}$, all the
arrows over $\mathfrak{p}$ incident into $s$ are of the same type. Otherwise
we say that $\mathfrak{p}$ is \emph{in}$\emph{homogeneous}$.
With this definition, we have the following result (see \cite{yao1}; also see \cite{huchen}).

\begin{theorem}\label{thm2}
\label{thm1} Let $\mathscr{A}=(S,i_0,\Sigma_k, t)$ be a finite $k$-automaton. Then we have the following two possibilities:
\begin{itemize}
\item[{(a)}] $\mathscr{A}$ is transparent if and only if all the paths on $\mathscr{A}$ are homogeneous;

\item[{(b)}] if $\mathscr{A}$ is not transparent, then $\Omega (\mathscr{A}%
)=1/2^{\ell -1}$, where $\ell \geqslant 2$ is the length of the shortest
inhomogeneous paths on $\mathscr{A}$.
\end{itemize}
\end{theorem}

\begin{proof}
First note that from the definitions, we have
\begin{eqnarray}\label{def}
\Omega (\mathscr{A})=\sup_{\eta \in \Sigma_k ^*}\inf_{o\in \Sigma_k ^{S}}\mathbf{d}(o(\mathscr{A}\eta ),\eta ).
\end{eqnarray}

(a) Suppose that $\mathscr{A}$ is transparent, \emph{i.e.}, $\Omega (\mathscr{A})=0$. Let $\mathfrak{p}$ be a path on $\mathscr{A}$. Then there exists $\eta\in \Sigma_k^*$ such that  $\mathfrak{p}= \mathfrak{p}_{\eta}$, where $\mathfrak{p}_{\eta}$
is the path determined by $\eta$, defined in Formula (\ref{eq0}). By Formula (\ref{def}) and the fact that $\Sigma_k ^{S}$ is finite, we can find $o_{\mathfrak{p}}\in \Sigma_k^{S}$ such that $o_{\mathfrak{p}}(\mathscr{A}\eta)=\eta$. Thus for each
vertex $s$ of the path $\mathfrak{p}$, all the arrows incident into $s$ over $\mathfrak{p}$ are labelled by $o_{\mathfrak{p}}(s)$. So $\mathfrak{p}$ is
homogeneous.

Conversely, assume that all the paths on $\mathscr{A}$ are
homogeneous. Take $\eta\in \Sigma_k^*$ and fix $\tau \in
\Sigma_k$. Then the path $\mathfrak{p}_{\eta}$ is homogeneous.
For all $s\in S$, if $s$ is not a vertex of $\mathfrak{p}_{\eta}$, then
we set $o_{\eta}(s):=\tau $. Otherwise we assign to $o_{\eta}(s)$ the label of the arrows incident into $s$ over $\mathfrak{p}_{\eta}$.
Then $o_{\eta}\in \Sigma_k ^{S}$ and $o_{\eta}(\mathscr{A}\eta)=\eta$. So
$\Omega (\mathscr{A})=0$, and $\mathscr{A}$ is transparent.

(b) If $\mathscr{A}$ is not transparent, it contains
some inhomogeneous paths. Since $\mathscr{A}$ is accessible, it suffices to show that for all $\eta\in \Sigma_k^*$ with $\mathfrak{p}_{\eta}$
inhomogeneous, we have
\begin{equation*}
\inf_{o\in \Sigma_k ^{S}}\mathbf{d}(o(\mathscr{A}\eta ),\eta )=2^{-h},
\end{equation*}
where $h\geqslant 1$ is the length of the longest homogeneous part of $\mathfrak{p}_{\eta}$.

Write $\eta=(\eta(j))_{0\leqslant j\leqslant m}$. Then $m\geqslant h$, and we can find an integer $j\ (0\leqslant j<h)$ such that
$\eta(j)\neq \eta(h)$, but $t(i_0, \eta[0,j])=t(i_0, \eta[0,h])$. So for all $o\in \Sigma_k^{S}$, we have
\begin{equation*}
o(t(i_0, \eta[0,j]))\neq \eta(j)\text{ or }o(t(i_0, \eta[0,h]))\neq \eta(h),
\end{equation*}
from which we obtain $\mathbf{d}(o(\mathscr{A}\eta),\eta)\geqslant2^{-h}$.

To end the proof, it remains to find  $o_{\eta}\in \Sigma_k^{S}$
such that%
\begin{equation}\label{eq2}
\mathbf{d}(o(\mathscr{A}\eta),\eta)=2^{-h}.
\end{equation}
For all $s\in S$, put $o_{\eta}(s):=\eta(l)$ if $s=t(i_0, \eta[0,l])$ for some integer $l\ (0\leqslant l< h)$. 
Otherwise put $o_{\eta}(s):=\tau $, where as above $\tau $ is a prefixed element of $\Sigma_k $.
From the definition of $h$, we deduce at once
\begin{equation*}
o_{\eta}(\mathscr{A}\eta)[0,h-1]=\eta[0,h-1],\text{
and }o_{\eta}(\mathscr{A}\eta)(h)\neq \eta(h),
\end{equation*}%
and the desired equality (\ref{eq2}) comes.
\end{proof}

Let $\mathscr{A}=(S,i_0,\Sigma_k, t)$ be a finite $k$-automaton, and $s$
a state of $\mathscr{A}$. Since $\mathscr{A}$ is supposed to be accessible, thus $s$ must have some incident arrows. We call $s$ a
\emph{homogeneous state} of type $\sigma\ (\sigma \in \Sigma_k )$ if over
the graph of $\mathscr{A}$, all the incident arrows into $s$ are of type $\sigma $.
Otherwise we call it {\em inhomogeneous}. Note here that inhomogeneous paths are determined by inhomogeneous states,
from which we can compute the opacity by Theorem \ref{thm1}.

A finite $k$-automaton is called \emph{homogeneous} if all its states  are
homogeneous. The following result tells us that transparent automata and homogeneous automata are tightly
related.

\begin{corollary}\label{cor1}
Let $\mathscr{A}=(S,i_0,\Sigma_k, t)$ be a finite $k$-automaton. If $\mathscr{A}$ is homogeneous, then it is transparent. The
converse also holds if $\mathscr{A}$ is strictly accessible.
\end{corollary}

\begin{proof}

The first part comes directly from Theorem \ref{thm1}. 
Indeed we can do better. For all $s\in S$, since $s$ is homogeneous, 
we let $o(s)$ denote the type  of $s$.
Thus we obtain $o\in \Sigma_k ^{S}$ such 
that $o(\mathscr{A}\sigma)=\sigma$, for all 
$\sigma\in \Sigma_k^*$.

Now assume that the finite automaton 
$\mathscr{A}$ is transparent and
strictly accessible. Let $s\in S$. 
If $(s_{1},\sigma_{1},s)$ and $(s_{2},\sigma _{2},s)$ 
are edges of $\mathscr{A}$, by the strict
accessibility, we can find a path $\mathfrak{p}$
on $\mathscr{A}$ which contains these two edges. But 
$\mathscr{A}$ is transparent, so $\mathfrak{p}$ is 
homogeneous, and then $\sigma_{1}=\sigma _{2}$, 
\emph{i.e.}, $s$ is homogeneous. 
\end{proof}

Likewise we also have the following characterization 
of opaque automata.

\begin{corollary}\label{cor2}
A finite $k$-automaton $\mathscr{A}=(S,i_0,\Sigma_k ,t)$ is opaque if and only if there exist $\sigma_1,\sigma_2\in\Sigma_k$ such that $\sigma _{1}\neq \sigma _{2}$,
and $t(i_0, \sigma_1)=t(i_0, \sigma_1\sigma_2)$.
\end{corollary}

\begin{proof}
It suffices to apply  Theorem \ref{thm2}.(b) 
with~$\ell=2$.
\end{proof}

\noindent {\bf Remark.} {Let $u=(u(n))_{n\geqslant 0}$ be a sequence generated by a finite $k$-automaton without output or with injective output function.
Then by Corollary~\ref{cor2}, we obtain that $u$ is opaque if and only there exist $\sigma_1,\sigma_2\in\Sigma_k$ such that $\sigma _{1}\neq \sigma _{2}$,
and $u(\sigma_1)=u(\sigma_1k+\sigma_2)$. See Example \ref{morse} and Example \ref{3-automaton}.}

\section{Some examples}

In the section, we compute the opacity 
of some classical automatic sequences.

\begin{example}\label{constants}
 {\em (constants)} The one-state automaton is the intrinsic finite $k$-automaton of all constant sequences. 
 Thus all constant sequences are opaque. Indeed they lose all the information of the words defined by the $k$-ary expansion of $n$'s.
\end{example}

\begin{example}\label{periodic}
{\em ($2$-periodic)} The identity automaton is the intrinsic finite $2$-automaton of purely $2$-periodic 
sequences. This automaton is homogeneous, hence  all purely $2$-periodic sequences are transparent.
\end{example}

\begin{example}\label{morse}
{\em (Thue-Morse)} The Thue-Morse sequence $(u(n))_{n\geqslant 0}$ in $\{0,1\}$ satisfies
\begin{eqnarray*}
u(0)=0,\ u(2n)=u(n),\ \text{and}\ u(2n+1)=1-u(n)\quad (n\geqslant 0).
\end{eqnarray*}
Its intrinsic finite $2$-automaton is given by Figure \ref{thue-morse} (see for example \cite[p.\,174]{AS}).
\vskip -0.3cm
\begin{figure}[ht]
    
\begin{tikzpicture}[shorten >=1pt,node distance=2cm,auto,initial text=]
 
  \node[initial, state]  (A)              {$A$};
  \node[state]  (B) [right=of A] {$B$};
 
  
  \path[->] (A) edge [in=60,out=120,loop] node[above] {$0$} (A)
             (A) edge [bend left, above]   node {$1$} (B)
             (B) edge [bend left, below]   node {$1$} (A)
            (B) edge [in=60,out=120,loop] node[above] {$0$} (B) ; 
\end{tikzpicture}

\centerline{Output function: $o(A) = 0$, $o(B) = 1$}
\vskip -0.2cm
\caption{Thue-Morse\ Automaton $\mathscr{A}_{tm}$} \label{thue-morse}
\end{figure}
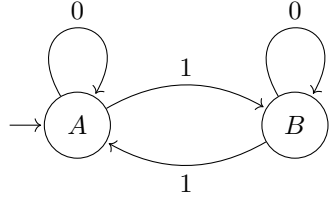

\noindent By Corollary \ref{cor2}, the Thue-Morse sequence is opaque, for $t(A,1)=B=t(A,10)$.
This result also comes from the Remark after Corollary \ref{cor2}, since $u(1)=u(2)$.
\end{example}

\begin{example}\label{period-doubling}
{\em (Period-doubling)} For all integers $n\geqslant 0$, define
\begin{eqnarray*}
u(n)=v_2(n+1)\pmod 2,
\end{eqnarray*}
where $v_2(n+1)$ is the greatest integer $r\geqslant 0$
such that $2^r$ divides $n+1$. Then we call $u_{pd}=(u(n))_{n\geqslant 0}$ the period-doubling sequence,
whose intrinsic finite $2$-automaton is given by 
Figure~\ref{PD} (see for example \cite[p.\,176]{AS}).
\vskip -0.3cm
\begin{figure}[ht]
    
\begin{tikzpicture}[shorten >=1pt,node distance=2cm,auto,initial text=]
 
  \node[initial,state]  (A)              {$A$};
  \node[state]  (B) [right=of A] {$B$};

  
  \path[->] (A) edge [in=60,out=120,loop] node[above] {$0$} (A)
             (A) edge [bend left, above]   node {$1$} (B)
             (B) edge [bend left, below]   node {$0,1$} (A)
            ; 
\end{tikzpicture}

\centerline{Output function: $o(A) = 0$, $o(B) = 1$}
\vskip -0.2cm
\caption{Period-doubling automaton $\mathscr{A}_{pd}$}\label{PD}
\end{figure}
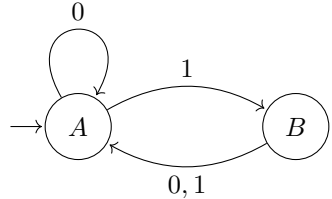

\noindent The initial state $A$ is the only inhomogeneous state of $\mathscr{A}_{pd}$. Thus by Theorem \ref{thm1}, we have $\Omega(\mathscr{A}_{pd})=1/4$ (the shortest inhomogeneous path is given by $\sigma=011$),
thus the opacity complexity of the period-doubling sequence equals $1/2$.
\end{example}

\begin{example}\label{golay-shapiro-rudin}
{\em (Golay-Shapiro)} The Golay-Shapiro(-Rudin) sequence 
$(u(n))_{n\geqslant 0}$ satisfies (for all integers $n\geqslant 0$)
\begin{eqnarray*}
u(0)=1,\ u(2n)=u(n),\ u(4n+1)=u(n),\ 
\text{and}\ u(4n+3)=-u(2n+1).
\end{eqnarray*}
Its intrinsic finite $2$-automaton is given in
Figure~\ref{RS} (see for example \cite[p.\,154]{AS}).

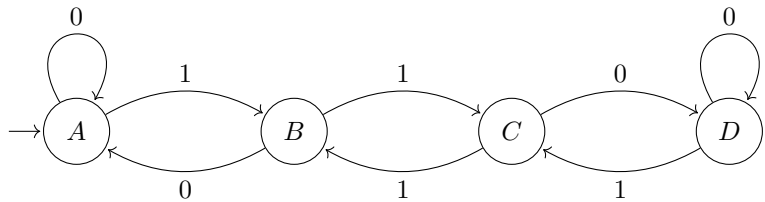
\begin{figure}[ht]
    
\begin{tikzpicture}[shorten >=1pt,node distance=2cm,auto,initial text=]
 
  \node[initial,state]  (A)              {$A$};
  \node[state]  (B) [right=of A] {$B$};
  \node[state]  (C) [right=of B] {$C$};
  \node[state]  (D) [right=of C] {$D$};
  \path[->] (A) edge [in=60,out=120,loop] node[above] {$0$} (A)
             (A) edge [bend left, above]   node {$1$} (B)
             (B) edge [bend left, below]   node {$0$} (A)
             (B) edge [bend left, above]   node {$1$} (C)
             (C) edge [bend left, below]   node {$1$} (B)
             (C) edge [bend left, above]   node {$0$} (D)
             (D) edge [bend left, below]   node {$1$} (C)
            (D) edge [in=60,out=120,loop] node[above] {$0$} (D) ; 
\end{tikzpicture}

\centerline{Output function: $o(A) = o(B) = 1$, 
$o(C) = o(D) =-1$.}

\caption{Golay-Shapiro(-Rudin) automaton $\mathscr{A}_{gsr}$}\label{RS}

\end{figure}

\noindent The $2$-automaton in Figure~\ref{RS} is 
homogeneous, hence by Corollary \ref{cor1}, the 
Golay-Shapiro(-Rudin) sequence is transparent.
\end{example}

\begin{example}\label{paperfolding}
{\em (Paperfolding)} The paperfolding sequence  $(u(n))_{n\geqslant 0}$ satisfies
\begin{eqnarray*}
u(2^n)=1,\ \text{and}\ u(2^n(2m+1))=(-1)^m\ (\forall n,m\geqslant 0).
\end{eqnarray*}
Its intrinsic finite $2$-automaton is given by 
Figure~\ref{pf} (see for example \cite[p.\,156]{AS}).

\begin{figure}[ht]
    
\begin{tikzpicture}[shorten >=1pt,node distance=2cm,auto,initial text=]
 
  \node[initial, state]  (A)              {$A$};
  \node[state]  (B) [right=of A] {$B$};
  \node[state]  (D) [below=of A] {$D$};
  \node[state]  (C) [below=of B] {$C$};
  \path[->] (A) edge [in=60,out=120,loop] node[above] {$0$} (A)
             (A) edge [bend left, above]   node {$1$} (B)
             (B) edge [bend left, below]   node {$0$} (A)
             (D) edge [in=-120,out=-60,loop] node[below] {$0$} (D) 
             (C) edge [in=-120,out=-60,loop] node[below] {$1$} (C)
             (D) edge node [left] {$1 \ $} (B)
             (B) edge node [right] {$1$} (C)
             (C) edge node [below] {$0$} (D); 
\end{tikzpicture}

\centerline{Output function: $o(A) = 
o(B) =1$, $ o(C) =o(D) = -1$.}

\caption{Paperfolding automaton $\mathscr{A}_{pf}$}\label{pf}
\end{figure}
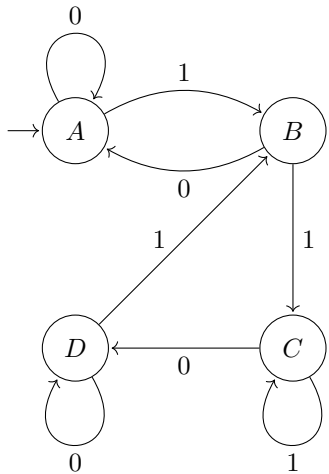

\noindent The $2$-automaton in Figure~\ref{pf} is 
homogeneous. Hence, by Corollary \ref{cor1}, the 
paperfolding sequence is transparent.
\end{example}

\begin{example}\label{baum-sweet}
{\em (Baum-Sweet)}  The Baum-Sweet sequence  $(u(n))_{n\geqslant 0}$ satisfies
\begin{eqnarray*}
u(0)=1,\ u(2n+1)=u(n),\ u(4n)=u(n),\ \text{and}\ u(4n+2)=0\ (n\geqslant 0).
\end{eqnarray*}
Its intrinsic finite $2$-automaton is given 
in Figure~\ref{bs}
(see for example \cite[p.\,157]{AS}).

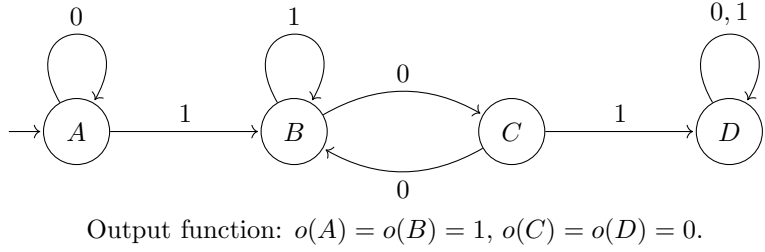
\begin{figure}[ht]
    
\begin{tikzpicture}[shorten >=1pt,node distance=2cm,auto,initial text=]
 
  \node[initial,state]  (A)              {$A$};
  \node[state]  (B) [right=of A] {$B$};
  \node[state]  (C) [right=of B] {$C$};
  \node[state]  (D) [right=of C] {$D$};
 
  
  \path[->] (A) edge [in=60,out=120,loop] node[above] {$0$} (A)
                edge                node {$1$} (B)
            (B) edge [in=60,out=120,loop] node[above] {$1$} (B)      
                edge [bend left, above]   node {$0$} (C)
            (C) edge              node {$1$} (D)
                edge [bend left, below] node{$0$} (B)
            (D) edge [in=60,out=120,loop] node[above]{$0,1$} (D);
\end{tikzpicture}

\centerline{Output function: $o(A) = o(B) = 1$, 
$o(C) = o(D) = 0$.}

\caption{Baum-Sweet automaton $\mathscr{A}_{bs}$}\label{bs}
\end{figure}

\noindent The inhomogeneous states of $\mathscr{A}_{bs}$ are $B$ and $D$.  Hence by Theorem \ref{thm1}, we have 
$\Omega(\mathscr{A}_{bs})=1/4$ (the shortest 
inhomogeneous path is given by $\sigma=100$),
thus the opacity complexity of the Baum-Sweet 
sequence equals $1/2$. It is worthy to point out here
that the Baum-Sweet sequence is opaque for the opacity
complexity attached to $\mathbf{d}_2$ (see \cite{allyao}).
As we have already indicated above, this is due to the fact
that the theory of opacity complexity heavily relies
on the prefixed comparison method.
\end{example}

\begin{example}\label{hanoi}
{\em (Tower of Hanoi)} 
The Tower of Hanoi sequence is a sequence on $6$
symbols, obtained from the classical Hanoi puzzle,
with ``an infinite number of disks''. Its intrinsic 
finite $2$-automaton is given in Figure~\ref{han}
(see for example \cite[p.\,181]{AS}).

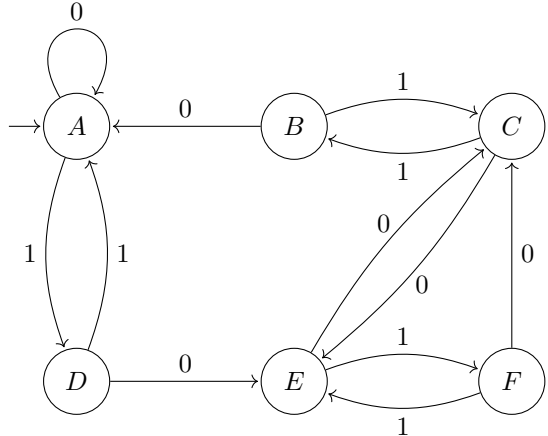
\begin{figure}[ht]
    
\begin{tikzpicture}[shorten >=1pt,node distance=2cm,auto,initial text=]
 
  \node[initial,state]  (A)              {$A$};
  \node[state]  (B) [right=of A] {$B$};
  \node[state]  (C) [right=of B] {$C$};
  \node[state]  (D) [below=2.5cm of A] {$D$};
  \node[state]  (E) [below=2.5cm of B] {$E$};
  \node[state]  (F) [below=2.5cm of C] {$F$};
  \path[->] (A) edge [in=60,out=120,loop] node[above] {$0$} (A)
             (B) edge node [above] {$0$} (A)
             (F) edge node [right] {$0$} (C)
             (B) edge [bend left=20, above]   node {$1$} (C)
             (C) edge [bend left=20, below]   node {$1$} (B)
             (E) edge [bend left=20, above]   node {$1$} (F)
             (F) edge [bend left=20, below]   node {$1$} (E)
             (D) edge node [above] {$0$} (E)
             (E) edge [bend left=10, above] node[above] {$0 \ $} (C)
             (C) edge [bend left=10, above] node[below] {$\ 0$} (E)
             (A) edge [bend right=20, left]   node {$1$} (D)
             (D) edge [bend right=20, right]   node {$1$} (A); 
\end{tikzpicture}

$$\alignedat6
\ \ \ \ \text{Output function: } \ 
&o(A) &&= a, \ &&o(B) &&=  \overline{a}, \ 
 &&o(C) &&= c, \\ 
&o(D) &&=  \overline{c}, \ &&o(E) &&= b, \
 &&o(F) &&= \overline{b}. \\
\endalignedat
$$

\caption{Tower of Hanoi Automaton $\mathscr{A}_{th}$}\label{han}
\end{figure}

\noindent The inhomogeneous states of $\mathscr{A}_{th}$ are $A$, $C$, and $E$. Hence by Theorem \ref{thm1}, we have $\Omega(\mathscr{A}_{th})=1/4$ 
(the shortest inhomogeneous path is given by $\sigma=011$),
thus the opacity complexity of the the Tower of Hanoi sequence equals $1/2$.
\end{example}

\begin{example}\label{3-automaton}
{\em (A $3$-automatic sequence)}
An example of $3$-automatic sequences, similar to the
Thue-Morse sequence, is the sequence $(z(n))_{n \geqslant  0}$
defined by: $z(n)$ is the sum, reduced modulo $3$, of 
the ternary digits of the integer $n$. Its intrinsic 
finite $3$-automaton is given in Figure~\ref{ter}.

\begin{figure}[ht]
    
\begin{tikzpicture}[shorten >=1pt,node distance=4cm,auto,initial text=]
 
  \node[initial,state]  (A) {$A$};
  \node[state]  (B) [right=of A] {$B$};
  \node[state]  (C) [below=3.5cm of A] {$C$};
  \path[->] (A) edge [in=60,out=120,loop] node[above] {$0$} (A)
  (B) edge [in=60,out=120,loop] node[above] {$0$} (B)
  (C) edge[loop, out=240, in=300, looseness=8] node[below]{$0$} (C)
             (A) edge [bend left=18, above] node {$1$} (B)
             (B) edge [bend left=18, below]   
             node {$2\ $} (A)
             (B) edge [bend left=15, right]   
             node {$\ 1$} (C)
             (C) edge [bend left=15, left]   node {$2\ $} (B)
             (C) edge [bend left=15, left] node {$1$} (A)
             (A) edge [bend left=15, right] node {$2$} (C);
\end{tikzpicture}

Output function: \ 
$o(A) = 0$, $o(B) = 1$,  
$o(C) = 2$.

\caption{Ternary sum of digits modulo $3$  Automaton $\mathscr{A}_{ter}$}\label{ter}
\end{figure}
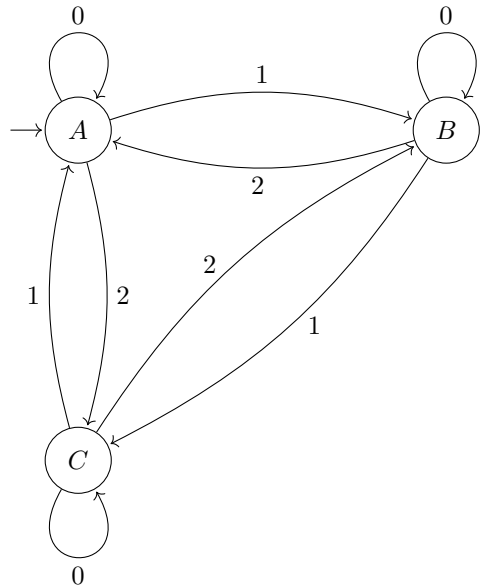
\noindent By Corollary \ref{cor2}, this $3$-automatic sequence is opaque, 
for $t(A,1)=B=t(B,10)$, just as for the Thue-Morse sequence. This result also comes from the Remark after Corollary \ref{cor2}, since $z(1)=z(2)$.
\end{example}

\section{Further study}

In the work in preparation \cite{allyao}, we shall continue our above study to consider opacity complexity attached to 
the comparison method $\mathbf{d}_2$, and discuss its various properties.
\medskip

\noindent \textbf{Acknowledgments.} Jia-Yan Yao would like to thank 
heartily the National Natural Science Foundation of China (Grants 
No.~12231013 and No.~11871295) for partial financial support.


\begin{thebibliography}{99}

\bibitem{AS} J.-P. Allouche and J. Shallit, \emph{Automatic Sequences.
Theory, Applications, Generalizations.} Cambridge University Press,
Cambridge (2003).

\bibitem{allyao} J.-P. Allouche and J.-Y. Yao, {\em Opacity complexity of automatic sequences. The quadratic case.} In preparation (2023).

\bibitem{bar2021} R. J. Barcelos and J. C. Basilio, {\em Enforcing current-state opacity through shuffle and deletions of event observations.}
Automatica J. IFAC \textbf{133} (2021), Paper No. 109836, 15~pp.

\bibitem{berard2015} B. B\'erard, K. Chatterjee, and N. Sznajder, {\em Probabilistic opacity for Markov decision processes.} 
Inform. Process. Lett. \textbf{115} (2015), 52--59.


\bibitem{liardet} A. Broglio and P. Liardet, {\em
Predictions with automata.} Symbolic dynamics and its applications 
(New Haven, CT, 1991), 111--124, Contemp. Math. \textbf{135}. 
Amer. Math. Soc., Providence, RI (1992).

\bibitem{carton1} O. Carton, {\em Preservation of normality by 
unambiguous transducers.} Developments in language theory, 90--101,
Lecture Notes in Comput. Sci. \textbf{13257}. Springer (2022).

\bibitem{carton2} O. Carton and E. Orduna, {\em Preservation of 
normality by transducers.} Inform. and Comput. \textbf{282} (2022), 
Paper No.~104650, 9~pp.

\bibitem{cheyao} G.-L. Chen and J.-Y. Yao, \emph{Characterization of 
opaque automata.} Discrete Math. \textbf{247} (2002), 65--78.

\bibitem{con} J. H. Conway, \emph{Regular Algebra and Finite Machines.}
Chapman and Hall Ltd. (1971).

\bibitem{huchen} H. Hu and X. M. Chen, {\em Opacity of composed finite automaton systems.} (Chinese)
J. Math. (Wuhan) \textbf{28} (2008), 343--348.

\bibitem{eil} S. Eilenberg, \emph{Automata, Languages and Machines.} 
Vol. A. Academic Press. New York (1974).

\bibitem{en} J. Endrullis, J. Karhum\"aki, J. W. Klop, and A. Saarela, 
{\em Degrees of infinite words, polynomials and atoms.} 
Internat. J. Found. Comput. Sci. \textbf{29} (2018), 825--843.

\bibitem{hem2021} L. A. Hemaspaandra and D. E. Narv\'aez,
{\em The opacity of backbones.}
Inform. and Comput. \textbf{281} (2021), Paper No. 104772, 10~pp.


\bibitem{kam} T. Kamae and M. Mend\`{e}s France, \emph{A continuous 
family of automata: the Ising automata.} Ann. Inst. H. Poincar\'{e}, Phys. Th\'{e}or. \textbf{64} (1996), 349--372.

\bibitem{laf2018} S. Lafortune, F. Lin, C. N. Hadjicostis, {\em On the history of diagnosability and opacity in discrete event systems.}
Annu. Rev. Control \textbf{45} (2018), 257--266.


\bibitem{Li2023} X. Li, C. N. Hadjicostis, and Z. Li, {\em Reduced-complexity verification for K-step and
infinite-step opacity in discrete event systems.} (2023), https://arxiv.org/abs/2310.11825.


\bibitem{men} M. Mend\`{e}s France, \emph{Opacity of an automaton.
Application to the inhomogeneous Ising chain.} Comm. Math. Phys. 
\textbf{139} (1991), 341--352.

\bibitem{vil} G. Viliam, P. Dominika, and S. Alexander, \emph{State complexity of binary coded regular languages.} Theoret. Comput. Sci. \textbf{990} (2024), Paper no. 114399. 


\bibitem{win2022}A. Wintenberg, M. Blischke, S. Lafortune, and N. Ozay,
{\em A general language-based framework for specifying and verifying notions of opacity.}
Discrete Event Dyn. Syst. \textbf{32} (2022), 253--289.

\bibitem{yang2021} J. Yang, W. Deng, D. Qiu, and C. Jiang,
{\em Opacity of networked discrete event systems.}
Inform. Sci. \textbf{543} (2021), 328--344.

\bibitem{yao} J.-Y. Yao, Contribution \`{a} l'\'{e}tude des automates 
finis (Th\`{e}se). Universit\'{e} Bordeaux I (1996).

\bibitem{yao1} J.-Y. Yao, \emph{Opacit\'{e}s des automates finis.} Discrete
Math. \textbf{202} (1999), 279--298.

\bibitem{yao2} J.-Y. Yao, \emph{Opacity of a finite automaton, method 
of calculation and the Ising chain.} Discrete Appl. Math. \textbf{125} 
(2003), 289--318.

\bibitem{yao3} J.-Y. Yao, \emph{Some properties of Ising automata.} 
Theoret. Comput. Sci. \textbf{314} (2004), 251--279.
\vskip 15pt

Jean-Paul ALLOUCHE

CNRS, IMJ-PRG, SORBONNE

4 Place Jussieu

F-75252 Paris Cedex 05

France

E-mail: jean-paul.allouche@imj-prg.fr
\vskip 10pt

Jia-Yan YAO

Department of Mathematics

Tsinghua University

Beijing 100084

People's Republic of China

E-mail: jyyao@mail.tsinghua.edu.cn
\end{thebibliography}
\end{document}